\documentclass[3p,10pt]{elsarticle}

\usepackage{amssymb,amsthm}
\usepackage{latexsym,amsmath,subfigure,color,siunitx,tikz}

\usepackage{lineno}

\journal{}

\newcommand{\eps}{\varepsilon}
\newcommand{\epsa}{\eps_\mathrm{a}}
\newcommand{\epsb}{\eps_\mathrm{b}}
\newcommand{\mua}{\mu_\mathrm{a}}
\newcommand{\mub}{\mu_\mathrm{b}}
\newcommand{\sigmaa}{\sigma_\mathrm{a}}
\newcommand{\sigmab}{\sigma_\mathrm{b}}
\newcommand{\ka}{k_\mathrm{aw}}
\newcommand{\kb}{k_\mathrm{bw}}
\newcommand{\set}[1]{\left\{#1\right\}}

\newcommand{\rmd}{\mathrm{d}}

\newcommand{\mf}{{\mathbf{f}}}

\newcommand{\mE}{{\mathbf{E}}}
\newcommand{\mH}{{\mathbf{H}}}
\newcommand{\mU}{{\mathbf{U}}}
\newcommand{\mV}{{\mathbf{V}}}
\newcommand{\mW}{{\mathbf{W}}}
\newcommand{\ma}{{\mathbf{a}}}
\newcommand{\mr}{{\mathbf{r}}}
\newcommand{\vt}{{\boldsymbol{\theta}}}

\DeclareMathOperator*{\inc}{inc}
\DeclareMathOperator*{\tot}{tot}

\DeclareMathOperator*{\scat}{scat}
\DeclareMathOperator*{\noise}{noise}

\theoremstyle{plain}
\newtheorem{theorem}{Theorem}[section]

\theoremstyle{remark}

\newtheorem{example}{Example}[section]

\begin{document}

\begin{frontmatter}



\title{Application of MUSIC-type imaging for anomaly detection without background information}

\author{Won-Kwang Park}
\ead{parkwk@kookmin.ac.kr}
\address{Department of Information Security, Cryptography, and Mathematics, Kookmin University, Seoul, 02707, Korea}

\begin{abstract}
It has been demonstrated that the MUltiple SIgnal Classification (MUSIC) algorithm is fast, stable, and effective for localizing small anomalies in microwave imaging. For the successful application of MUSIC, exact values of permittivity, conductivity, and permeability of the background must be known. If one of these values is unknown, it will fail to identify the location of an anomaly. However, to the best of our knowledge, no explanation of this failure has been provided yet. In this paper, we consider the application of MUSIC to the localization of a small anomaly from scattering parameter data when complete information of the background is not available. Thanks to the framework of the integral equation formulation for the scattering parameter data, an analytical expression of the MUSIC-type imaging function in terms of the infinite series of Bessel functions of integer order is derived. Based on the theoretical result, we confirm that the identification of a small anomaly is significantly affected by the applied values of permittivity and conductivity. However, fortunately, it is possible to recognize the anomaly if the applied value of conductivity is small. Simulation results with synthetic data are reported to demonstrate the theoretical result.
\end{abstract}


\end{frontmatter}






\section{Introduction}\label{sec:1}
Although the MUltiple SIgnal Classification (MUSIC) algorithm was developed for estimating the individual frequencies of multiple time-harmonic signals, it was successfully applied to an inverse scattering problem of localizing a set of point-like scatterers \cite{D}. From this pioneering research, MUSIC has been applied to various problems, for example, identification of arbitrarily shaped targets in inverse scattering problem \cite{AIL2,CZ,P-MUSIC7} as well as the microwave imaging  \cite{KCC,P-MUSIC6,SD}, detection of detecting internal corrosion \cite{AKKLV}, damage diagnosis on complex aircraft structures \cite{BYG,FZSY}, radar imagings \cite{CPPPDT,LWYL,ZZK}, impedance tomography \cite{H3}, ultrasound imaging \cite{LH2}, and medical imaging \cite{RSAAP,S2,SKLKLJC}.

Several studies have demonstrated that the MUSIC algorithm is fast, effective, and stable in both the inverse scattering problem and microwave imaging. However, for its successful application, one must \textcircled{1} discriminate nonzero singular values to determine the exact noise subspace and \textcircled{2} know a priori information of the background (exact values of background permittivity, conductivity, and permeability at a given frequency) to design an imaging function. Several studies \cite{GD,HSZ1,PL1,SMP,XXCT} have revealed certain properties of singular values and methods for appropriate threshold schemes. However, although some research has been performed on certain phenomena when an inaccurate frequency is applied (see \cite{P-MUSIC3,PP1,SRDCAR}, for instance), to the best of our knowledge it remains unexplained why one cannot localize a small anomaly when an inaccurate value of background permittivity or conductivity is applied. This provides the motivation for this study aimed at determining the effect of an applied inaccurate value of background permeability, permittivity, or conductivity.

The purpose of this paper is to establish a new mathematical theory of MUSIC in microwave imaging when the background information is unknown. To this end, we carefully explore the structure of MUSIC imaging function by constructing a relationship with infinite series of Bessel functions of integer order, antenna arrangement, and an applied inaccurate value of background wavenumber. This is based on the integral equation formula for the scattered-field $S-$parameter in the presence of a small anomaly and the structure of left-singular vector associated with the nonzero singular value of the scattering matrix. From the explored structure, we can explain that \textcircled{1} when an inaccurate value of background permeability or permittivity is applied, the identified location of the small anomaly is shifted in a specific direction, \textcircled{2} when an inaccurate value of background conductivity is applied, there is no shifting effect and it is possible to identify the location fairly precisely if the applied value is small, \textcircled{3} however it will be very difficult to recognize the existence of an anomaly if the applied value of conductivity is not small. To validate the theoretical results, various simulation results in the presence of single and multiple anomalies are presented.

This paper is organized as follows. In Section \ref{sec:2}, we briefly introduce the basic concept of scattering parameters in the presence of a small anomaly, introduce the imaging function of MUSIC. In Section \ref{sec:3}, the mathematical structure of the imaging function is explored by establishing a relationship with an infinite series of Bessel functions, antenna arrangement, and an applied inaccurate wavenumber. In Section \ref{sec:4}, we present various results of numerical simulations with synthetic data to confirm the theoretical results and discuss certain phenomena. In Section \ref{sec:5}, a short conclusion including future works is provided.

\section{Scattering parameter and the imaging function of MUSIC}\label{sec:2}
Let $D$ be a circle-like small anomaly with radius $\alpha$, location $\mr_\star$, permittivity $\eps_\star$, and conductivity $\sigma_\star$ at given angular frequency $\omega$. We set $D$ to be surrounded by a circular array of dipole antennas $\Lambda_n$, $n=1,2,\cdots,N$, with location $\ma_n$ and they are placed outside of the homogeneous region of interest (ROI) $\Omega$. In this paper, we assume that there exists no magnetic materials in $\Omega$ thus, the anomaly $D$ and the background $\Omega$ are characterized by the value of dielectric permittivity and electric conductivity at a given angular frequency $\omega=2\pi f$, where $f$ denotes the ordinary frequency measured in \texttt{hertz}. We denote $\epsb$ and $\sigmab$ as the permittivity and conductivity of $\Omega$, respectively, and $\eps_\star$ and $\sigma\star$ that satisfy
\begin{equation}\label{Condition}
\omega\epsb\gg\sigmab,\quad\text{and}\quad2\alpha\sqrt{\frac{\eps_\star}{\epsb}}<\text{wavelength},
\end{equation}
and set the value of magnetic permeability as a constant such that $\mu(\mr)\equiv\mub=\SI{1.257e-6}{\henry/\meter}$ for every $\mr\in\Omega$. With this, we denote $\kb$ be the background wavenumber that satisfies
\[\kb^2=\omega^2\mub\left(\epsb+i\frac{\sigmab}{\omega}\right)\]
and define the following piecewise permittivity and conductivity as $\eps(\mr)$ and $\sigma(\mr)$, respectively such that
\[\eps(\mr)=\left\{\begin{array}{rcl}
\smallskip\eps_\star & \text{for} & \mr\in D,\\
\epsb & \text{for} & \mr\in\Omega\backslash\overline{D},
\end{array}\right.
\quad\mbox{and}\quad
\sigma(\mr)=\left\{\begin{array}{rcl}
\sigma_\star & \mbox{for} & \mr\in D,\\
\sigmab & \mbox{for} & \mr\in\Omega\backslash\overline{D},
\end{array}\right.\]

Let $\mE_{\inc}(\kb,\mr,\ma_m)\in\mathbb{C}^{1\times3}$ be the incident electric field in $\Omega$ due to the point current density $\mathbf{J}$ at $\Lambda_m$ that satisfies
\[\left\{\begin{array}{l}
\smallskip\nabla\times\mE_{\inc}(\kb,\mr,\ma_m)=-i\omega\mub\mH_{\inc}(\kb,\mr,\ma_m),\\
\nabla\times\mH_{\inc}(\kb,\mr,\ma_m)=(\sigmab+i\omega\epsb)\mE_{\inc}(\kb,\mr,\ma_m),
\end{array}\right.\]
where ${\mH_{\inc}}\in\mathbb{C}^{1\times3}$ denotes the magnetic field. Analogously, let $\mE_{\tot}(\kb,\ma_n,\mr)\in\mathbb{C}^{1\times3}$ be the total electric field in the existence of $D$ measured at $\Lambda_n$ that satisfies
\[\left\{\begin{array}{l}
\smallskip\nabla\times\mE_{\tot}(\kb,\ma_n,\mr)=-i\omega\mub\mH_{\tot}(\kb,\ma_n,\mr),\\
\nabla\times\mH_{\tot}(\kb,\ma_n,\mr)=(\sigma(\mr)+i\omega\eps(\mr))\mE_{\tot}(\kb,\ma_n,\mr)
\end{array}\right.\]
with the transmission condition on the boundary of $D$.

Let $S_{\inc}(n,m)$ be the incident-field $S-$parameter, which is the scattering parameter without $D$ with transmitter number $m$ and receiver number $n$. Similarly, $S_{\tot}(n,m)$ be the total-field $S-$parameter, which is the scattering parameter in the presence of $D$ with transmitter number $m$ and receiver number $n$. In this paper, the measurement data to retrieve $D$ is the scattered-field $S-$parameter with transmitter number $m$ and receiver number $n$ denoted by $S_{\scat}(n,m)=S_{\tot}(n,m)-S_{\inc}(n,m)$. Then, on the basis of \cite{HSM2}, $S_{\scat}(n,m)$ can be represented by the following integral equation
\begin{equation}\label{IntegralExpression}
S_{\scat}(n,m)=\frac{i\kb^2}{4\omega\mub}\int_\Omega\left(\frac{\eps(\mr')-\epsb}{\epsb}+i\frac{\sigma(\mr')-\sigmab}{\omega\epsb}\right)\mE_{\inc}(\kb,\mr',\ma_m)\cdot\mE_{\tot}(\kb,\ma_n,\mr')\rmd\mr'.
\end{equation}
Notice that the exact expression of $\mE_{\tot}(\kb,\ma_n,\mr')$ is unknown, it is very hard to apply \eqref{IntegralExpression} to design MUSIC algorithm. Since the condition \eqref{Condition} holds, it is possible to apply the Born approximation to \eqref{IntegralExpression}. Then, $S_{\scat}(n,m)$ can be approximated as
\[S_{\scat}(n,m)\approx\frac{i\kb^2}{4\omega\mub}\int_D\left(\frac{\eps_\star-\epsb}{\epsb}+i\frac{\sigma_\star-\sigmab}{\omega\epsb}\right)\mE_{\inc}(\kb,\mr',\ma_m)\cdot\mE_{\inc}(\kb,\ma_n,\mr')\rmd\mr'.\]
It is worth to emphasize that based on the simulation configuration \cite{KLKJS}, only the $z$-component of the field $\mE_{\inc}(\kb,\mr,\ma_n)$ can be measured at $\Lambda_n$. Hence, by denoting it as $u(\kb,\mr',\ma_m)$ and applying the mean-value theorem, $S_{\scat}(n,m)$ can be written as
\begin{equation}\label{ScatteringParameter}
S_{\scat}(n,m)\approx\frac{i\alpha^2\kb^2\pi}{4\omega\mub}\left(\frac{\eps_\star-\epsb}{\epsb}+i\frac{\sigma_\star-\sigmab}{\omega\epsb}\right)u(\kb,\ma_m,\mr_\star)u(\kb,\ma_n,\mr_\star).
\end{equation}

To introduce the imaging function of MUSIC algorithm, we perform the singular value decomposition for the scattering matrix
\[\mathbb{K}=\begin{bmatrix}
0&S_{\scat}(1,2)&\cdots&S_{\scat}(1,N)\\
S_{\scat}(2,1)&0&\cdots&S_{\scat}(2,N)\\
\vdots&\vdots&\ddots&\vdots\\
S_{\scat}(N,1)&S_{\scat}(N,2)&\cdots&0
\end{bmatrix}=\mathbb{UDV}^*\approx\tau_1\mU_1\mV_1^*,\]
where $\tau_1$ denotes the nonzero singular value, and $\mU_1$ and $\mV_1$ are the first left- and right-singular vectors of the scattering matrix, respectively. We refer to \cite{P-MUSIC6} why the diagonal elements of $\mathbb{K}$ are set to zero. With this, by denoting $\mathbb{I}$ as the $N\times N$ identity matrix, we can define projection operator $\mathbb{P}_{\noise}$ onto the noise subspace:
\[\mathbb{P}_{\noise}=\mathbb{I}-\mU_1\mU_1^*.\]
Then, based on the structure of the approximation \eqref{ScatteringParameter}, we introduce a unit test vector: for each $\mr\in\Omega$,
\begin{equation}\label{TestVector}
\mW(\kb,\mr)=\frac{\mf(\kb,\mr)}{|\mf(\kb,\mr)|},\quad\text{where}\quad\mf(\kb,\mr)=\Big[u(\kb,\ma_1,\mr_\star),u(\kb,\ma_2,\mr_\star),\ldots,u(\kb,\ma_N,\mr_\star)\Big]^T.
\end{equation}
Then, since
\[\mW(\kb,\mr)\in\text{Range}(\mathbb{K})\quad\text{if and only if}\quad\mr=\mr_\star\in D,\]
we can examine that $|\mathbb{P}_{\noise}(\mW(\kb,\mr_\star))|=0$ and by plotting the following imaging function of MUSIC
\[\mathfrak{F}(\kb,\mr)=\frac{1}{|\mathbb{P}_{\noise}(\mW(\kb,\mr))|},\quad\mr\in\Omega,\]
the location $\mr_\star\in D$ can be identified. We refer to \cite{P-MUSIC6,AK2,C,ZC} for detailed descriptions.

It is important that, to generate the test vector $\mW(\kb,\mr)$, on the basis of \eqref{TestVector}, the exact value of $\kb$ must be known. This means that exact values of $\omega$, $\mub$, $\epsb$, and $\sigmab$ must be known. However, their exact values are sometimes unknown because these values are significantly dependent on the frequency, temperature, and other factors. Now, we assume that the exact values of $\epsb$ and $\sigmab$ are unknown, and apply an alternative value $\ka$ instead of the true $\kb$. Correspondingly, we set a unit test vector $\mW(\ka,\mr)$ from \eqref{TestVector} and consider the imaging function of MUSIC
\begin{equation}\label{ImagingFunction}
\mathfrak{F}(\ka,\mr)=\frac{1}{|\mathbb{P}_{\noise}(\mW(\ka,\mr))|},\quad\mr\in\Omega.
\end{equation}

Notice that, the exact location of $D$ cannot be retrieved through the map of $\mathfrak{F}(\ka,\mr)$, but we can recognize the existence of an anomaly and the identified location is shifted in a specific direction. However, some phenomena exhibited in Section \ref{sec:4} cannot be explained yet.

\section{Theoretical result: structure of the imaging function with inaccurate wavenumber}\label{sec:3}
In this section, we explore the structure of the imaging function $\mathfrak{F}(\ka,\mr)$ to explain the theoretical reason of some phenomena. The result is following.

\begin{theorem}\label{Structure}Let $\vt_n=\ma_n/|\ma_n|=\ma_n/R=(\cos\theta_n,\sin\theta_n)$ and $\ka\mr-\kb\mr_\star=|\ka\mr-\kb\mr_\star|(\cos\phi,\sin\phi)$. {If $\ma_n$ satisfies $|\ma_n-\mr|\gg\set{1/4|\kb|,1/4|\ka|}$ for $n=1,2,\cdots,N$,} $\mathfrak{F}(\ka,\mr)$ can be represented as follows:
\begin{equation}\label{StructureImagingFunction}
\mathfrak{F}(\ka,\mr)\approx\frac{N^2-2N+1}{N^2-2N}
\bigg(1-\bigg|J_0(|\ka\mr-\kb\mr_\star|)+\frac{1}{N}\sum_{n=1}^{N}\sum_{q\in\mathbb{Z}_0}i^qJ_q(|\ka\mr-\kb\mr_\star|)e^{iq(\theta_n-\phi)}\bigg|^2\bigg)^{-1/2},
\end{equation}
where $J_s$ denotes the Bessel function of order $s$ and $\mathbb{Z}_0$ denotes the set of integer number except $0$.
\end{theorem}
\begin{proof}
Based on \cite{PKLS}, the incident field $u(\kb,\mr,\mr')$ can be written as
\begin{equation}\label{IncidentField}
u(\kb,\mr,\mr')=\frac{i}{4}H_0^{(2)}(\kb|\mr-\mr'|),\quad\mr\ne\mr',
\end{equation}
where $H_0^{(2)}$ denotes the Hankel function of order zero of the second kind. Since $|\ma_n-\mr|,|\ma_n-\mr_\star|\gg\set{1/4|\kb|,1/4|\ka|}$ for all $n=1,2,\cdots,N$, the following asymptotic forms of the Hankel function hold (see \cite[Theorem 2.5]{CK}, for instance)
\begin{equation}\label{AsymptoticHankel}
\frac{i}{4}H_0^{(2)}(\kb|\ma_n-\mr'|)\approx\frac{(-1+i)e^{-i\kb|\ma_n|}}{4\sqrt{\kb\pi|\ma_n|}}e^{i\kb\vt_n\cdot\mr'}\quad\text{and}\quad\frac{i}{4}H_0^{(2)}(\ka|\ma_n-\mr'|)\approx\frac{(-1+i)e^{-i\ka|\ma_n|}}{4\sqrt{\ka\pi|\ma_n|}}e^{i\kb\vt_n\cdot\mr'},
\end{equation}
the unit test vector $\mW(\ka,\mr)$ becomes
\begin{equation}\label{UnitTestVector}
\mW(\ka,\mr)\approx\frac{1}{\sqrt{N}}\bigg[e^{i\ka\vt_1\cdot\mr},e^{i\ka\vt_2\cdot\mr},\cdots,e^{i\ka\vt_N\cdot\mr}\bigg]^T,
\end{equation}
and the scattering matrix $\mathbb{K}$ can be written as
\[\mathbb{K}\approx\frac{\alpha^2\kb e^{-2i\kb R}}{32R\omega\mub}\left(\frac{\eps_\star-\epsb}{\epsb}+i\frac{\sigma_\star-\sigmab}{\omega\epsb}\right)\begin{bmatrix}
 0 & e^{ik(\vt_1+\vt_2)\cdot\mr_\star} & \cdots & e^{ik(\vt_1+\vt_N)\cdot\mr_\star}\\
e^{ik(\vt_2+\vt_2)\cdot\mr_\star} & 0 & \cdots & e^{ik(\vt_2+\vt_N)\cdot\mr_\star}\\
\vdots&\vdots&\ddots&\vdots\\
e^{ik(\vt_N+\vt_1)\cdot\mr_\star} & e^{ik(\vt_N+\vt_2)\cdot\mr_\star} & \cdots & 0\\
\end{bmatrix}.\]

Let us denote
\[\mathcal{O}=\frac{\eps_\star-\epsb}{\epsb}+i\frac{\sigma_\star-\sigmab}{\omega\epsb}\quad\text{and}\quad\mathbb{M}=\begin{bmatrix}
 0 & e^{ik(\vt_1+\vt_2)\cdot\mr_\star} & \cdots & e^{ik(\vt_1+\vt_N)\cdot\mr_\star}\\
e^{ik(\vt_2+\vt_2)\cdot\mr_\star} & 0 & \cdots & e^{ik(\vt_2+\vt_N)\cdot\mr_\star}\\
\vdots&\vdots&\ddots&\vdots\\
e^{ik(\vt_N+\vt_1)\cdot\mr_\star} & e^{ik(\vt_N+\vt_2)\cdot\mr_\star} & \cdots & 0\\
\end{bmatrix}.\]
Then, by performing an elementary calculus, we can examine that
\begin{align*}
\mathbb{M}\mathbb{M}^*&=\begin{bmatrix}
N-1 & (N-2)e^{ik(\vt_1-\vt_2)\cdot\mr_\star} & \cdots & (N-2)e^{ik(\vt_1-\vt_N)\cdot\mr_\star}\\
(N-2)e^{ik(\vt_2-\vt_1)\cdot\mr_\star} & N-1 & \cdots & (N-2)e^{ik(\vt_2-\vt_N)\cdot\mr_\star}\\
\vdots&\vdots&\ddots&\vdots\\
(N-2)e^{ik(\vt_N-\vt_1)\cdot\mr_\star} & (N-2)e^{ik(\vt_N-\vt_2)\cdot\mr_\star} & \cdots & N-1\\
\end{bmatrix}\\
&=\mathbb{I}+(N-2)\begin{bmatrix}
e^{ik(\vt_1-\vt_1)\cdot\mr_\star} & e^{ik(\vt_1-\vt_2)\cdot\mr_\star} & \cdots & e^{ik(\vt_1-\vt_N)\cdot\mr_\star}\\
e^{ik(\vt_2-\vt_1)\cdot\mr_\star} & e^{ik(\vt_2-\vt_2)\cdot\mr_\star} & \cdots & e^{ik(\vt_2-\vt_N)\cdot\mr_\star}\\
\vdots&\vdots&\ddots&\vdots\\
e^{ik(\vt_N-\vt_1)\cdot\mr_\star} & e^{ik(\vt_N-\vt_2)\cdot\mr_\star} & \cdots & e^{ik(\vt_N-\vt_N)\cdot\mr_\star}\\
\end{bmatrix}
\end{align*}
and correspondingly, we have
\begin{align*}
\mU_1\mU_1^*&=\frac{1}{|\tau_1|^2}\mathbb{KK}^*\approx\left|\frac{\alpha^2\mathcal{O}\kb}{32R\omega\mub\tau_1}\right|^2\mathbb{MM}^*\\
&=C\mathbb{I}+C(N-2)\begin{bmatrix}
e^{ik(\vt_1-\vt_1)\cdot\mr_\star} & e^{ik(\vt_1-\vt_2)\cdot\mr_\star} & \cdots & e^{ik(\vt_1-\vt_N)\cdot\mr_\star}\\
e^{ik(\vt_2-\vt_1)\cdot\mr_\star} & e^{ik(\vt_2-\vt_2)\cdot\mr_\star} & \cdots & e^{ik(\vt_2-\vt_N)\cdot\mr_\star}\\
\vdots&\vdots&\ddots&\vdots\\
e^{ik(\vt_N-\vt_1)\cdot\mr_\star} & e^{ik(\vt_N-\vt_2)\cdot\mr_\star} & \cdots & e^{ik(\vt_N-\vt_N)\cdot\mr_\star}\\
\end{bmatrix},\quad C=\left|\frac{\alpha^2\mathcal{O}\kb}{32R\omega\mub\tau_1}\right|^2\in\mathbb{R}.
\end{align*}

Since the following JacobiAnger expansion formula holds uniformly,
\begin{equation}\label{JacobiAnger}
e^{ix\cos\theta}=J_0(x)+\sum_{q\in\mathbb{Z}_0}i^qJ_q(x)e^{iq\theta},
\end{equation}
we can evaluate that
\begin{align}
\begin{aligned}\label{TermBessel}
\sum_{n=1}^{N}e^{i\vt_n\cdot(\ka\mr-\kb\mr_\star)}
&=\sum_{n=1}^{N}e^{i|\ka\mr-\kb\mr_\star|\cos(\theta_n-\phi)}\\
&=\bigg(NJ_0(|\ka\mr-\kb\mr_\star|)+\sum_{n=1}^{N}\sum_{q\in\mathbb{Z}_0}i^qJ_q(|\ka\mr-\kb\mr_\star|)e^{iq(\theta_n-\phi)}\bigg)\\
&=N\bigg(J_0(|\ka\mr-\kb\mr_\star|)+\frac{1}{N}\sum_{n=1}^{N}\sum_{q\in\mathbb{Z}_0}i^qJ_q(|\ka\mr-\kb\mr_\star|)e^{iq(\theta_n-\phi)}\bigg)\\
&:=N\bigg(J_0(|\ka\mr-\kb\mr_\star|)+\mathcal{E}(\ka\mr,\kb\mr_\star)\bigg).
\end{aligned}
\end{align}
With this, by applying \eqref{UnitTestVector} and \eqref{TermBessel}, we can evaluate
\begin{multline*}
\left(\mathbb{I}-\mU_1\mU_1^*\right)\mW(\ka,\mr)
\approx\frac{(1-C)}{\sqrt{N}}\begin{bmatrix}
\medskip e^{i\ka\vt_1\cdot\mr}\\
e^{i\ka\vt_2\cdot\mr}\\
\medskip\vdots\\
e^{i\ka\vt_N\cdot\mr}\end{bmatrix}\\
-C(N-2)\sqrt{N}\begin{bmatrix}
\medskip e^{i\kb\vt_1\cdot\mr_\star}\Big(J_0(|\ka\mr-\kb\mr_\star|)+\mathcal{E}(\ka\mr,\kb\mr_\star)\Big)\\
 e^{i\kb\vt_2\cdot\mr_\star}\Big(J_0(|\ka\mr-\kb\mr_\star|)+\mathcal{E}(\ka\mr,\kb\mr_\star)\Big)\\
\medskip\vdots\\
 e^{i\kb\vt_N\cdot\mr_\star}\Big(J_0(|\ka\mr-\kb\mr_\star|)+\mathcal{E}(\ka\mr,\kb\mr_\star)\Big)
\end{bmatrix}
\end{multline*}
and correspondingly,
\begin{align*}
|\mathbb{P}_{\noise}(\mW(\ka,\mr))|&=\Big(\mathbb{P}_{\noise}(\mW(\ka,\mr))\cdot\overline{\mathbb{P}_{\noise}(\mW(\ka,\mr))}\Big)^{1/2}\\
&=\bigg[\sum_{n=1}^{N}\bigg(\frac{(1-C)^2}{N}-(\Psi_1+\overline{\Psi}_1)+\Psi_2\overline\Psi_2\bigg)\bigg]^{1/2},\end{align*}
where
\begin{align*}
\Psi_1&=(1-C)C(N-2)e^{i\vt_n\cdot(\ka\mr-\kb\mr_\star)}\Big(\overline{J_0(|\ka\mr-\kb\mr_\star|)+\mathcal{E}(\ka\mr,\kb\mr_\star)}\Big)\\
\Psi_2&=C(N-2)Ne^{i\kb\vt_n\cdot\mr_\star}\Big(J_0(|\ka\mr-\kb\mr_\star|)+\mathcal{E}(\ka\mr,\kb\mr_\star)\Big).
\end{align*}

Applying \eqref{TermBessel} again, we can easily obtain that
\begin{align*}
&\sum_{n=1}^{N}(\Psi_1+\overline{\Psi}_1)=2(1-C)C(N-2)N\Big|J_0(|\ka\mr-\kb\mr_\star|)+\mathcal{E}(\ka\mr,\kb\mr_\star)\Big|^2\\
&\sum_{n=1}^{N}\Psi_2\overline\Psi_2=C^2(N-2)^2N^2\Big|J_0(|\ka\mr-\kb\mr_\star|)+\mathcal{E}(\ka\mr,\kb\mr_\star)\Big|^2.
\end{align*}
Therefore,
\begin{multline*}
|\mathbb{P}_{\noise}(\mW(\ka,\mr))|\approx\bigg((1-C)^2-2C(1-C)(N-2)N\Big|J_0(|\ka\mr-\kb\mr_\star|)+\mathcal{E}(\ka\mr,\kb\mr_\star)\Big|^2\\
+C^2(N-2)^2N^2\Big|J_0(|\ka\mr-\kb\mr_\star|)+\mathcal{E}(\ka\mr,\kb\mr_\star)\Big|^2\bigg)^{1/2}.
\end{multline*}
Finally, since $|\mathbb{P}_{\noise}(\mW(\ka,\mr))|=0$ and $|J_0(|\ka\mr-\kb\mr_\star|)+\mathcal{E}(\ka\mr,\kb\mr_\star)|=1$ when $\ka\mr=\kb\mr_\star$,
\[(1-C)^2-2C(1-C)(N-2)N+C^2(N-2)^2N^2=0\]
or equivalently, $\big((1-C)-C(N-2)N\big)^2=0$. Hence, $C(N-1)^2=1$ and correspondingly,
\begin{multline*}
|\mathbb{P}_{\noise}(\mW(\ka,\mr))|\approx|1-C|^2\left(1-\Big|J_0(|\ka\mr-\kb\mr_\star|)+\mathcal{E}(\ka\mr,\kb\mr_\star)\Big|^2\right)^{1/2}\\
=\frac{N^2-2N}{N^2-2N+1}\bigg(1-\bigg|J_0(|\ka\mr-\kb\mr_\star|)+\frac{1}{N}\sum_{n=1}^{N}\sum_{q\in\mathbb{Z}_0}i^qJ_q(|\ka\mr-\kb\mr_\star|)e^{iq(\theta_n-\phi)}\bigg|^2\bigg)^{1/2}.
\end{multline*}
With this, we can obtain the structure \eqref{StructureImagingFunction}.
\end{proof}

Based on the identified structure $\mathfrak{F}(\ka,\mr)$, we can say that the location $\mr=(\kb/\ka)\mr_\star$ will be identified instead of the true one $\mr_\star$ because $J_0(|\ka\mr-\kb\mr_\star|)=1$ and $\mathcal{E}(\mr)=0$ when $|\ka\mr-\kb\mr_\star|=0$. This is the reason why the identified location of the anomaly is shifted. Note that, if the anomaly is located at the origin, its location can be identified for any value $\ka$. Further properties will be discussed in the simulation results.

\section{Simulation results and discussions}\label{sec:4}
In this section, we present the results of simulation with synthetic data to check the theoretical result. To this end, $N = 16$ dipole antennas were used to transmit/receive signals at $f =\SI{1}{\giga\hertz}$ such that
\[\ma_n=\SI{0.09}{\meter}\left(\cos\frac{2n\pi}{N},\sin\frac{2n\pi}{N}\right),\quad n=1,2,\cdots,N.\]
The ROI $\Omega$ was set to be an interior of a circle with $(\epsb,\sigmab)=(20\eps_0,\SI{0.2}{\siemens/\meter})$ and radius $\SI{0.085}{\meter}$ centered at the origin. Here, $\eps_0=\SI{8.854e-12}{\farad/\meter}$ is the vacuum permittivity. For the anomaly, we select a small ball $D$ with $\mr_1=(\SI{0.01}{\meter},\SI{0.03}{\meter})$, $\alpha_1=\SI{0.01}{\meter}$, and $(\eps_1,\sigma_1)=(55\eps_0,\SI{1.2}{\siemens/\meter})$. For multiple anomalies, we select another small ball $D_2$ with $\mr_2=(-\SI{0.04}{\meter},-\SI{0.02}{\meter})$, $\alpha_2=\alpha_1$, and $(\eps_2,\sigma_2)=(45\eps_0,\SI{1.0}{\siemens/\meter})$. We refer to Figure \ref{IllustrationAnomalies} for an illustration of simulation configurations.

\begin{figure}[h]
\begin{center}
\begin{tikzpicture}[scale=2.4]
\foreach \alpha in {0,22.5,...,337.5}
{\draw[green,fill=green] ({cos(\alpha)},{sin(\alpha)}) circle (0.05cm);
\draw[black,fill=black] ({cos(\alpha)},{sin(\alpha)}) circle (0.02cm);}
\node at (0,0) {\footnotesize$(\epsb,\sigmab)=(20\eps_0,\SI{0.2}{\siemens/\meter})$};
\end{tikzpicture}\hfill
\begin{tikzpicture}[scale=2.4]
\foreach \alpha in {0,22.5,...,337.5}
{\draw[green,fill=green] ({cos(\alpha)},{sin(\alpha)}) circle (0.05cm);
\draw[black,fill=black] ({cos(\alpha)},{sin(\alpha)}) circle (0.02cm);}
\draw[violet,fill=violet] (0.1,0.3) circle (0.1cm);
\node at (0,0.1) {\footnotesize$(\eps_1,\sigma_1)=(55\eps_0,\SI{1.2}{\siemens/\meter})$};
\end{tikzpicture}\hfill
\begin{tikzpicture}[scale=2.4]
\foreach \alpha in {0,22.5,...,337.5}
{\draw[green,fill=green] ({cos(\alpha)},{sin(\alpha)}) circle (0.05cm);
\draw[black,fill=black] ({cos(\alpha)},{sin(\alpha)}) circle (0.02cm);}
\draw[violet,fill=violet] (0.1,0.3) circle (0.1cm);
\node at (0,0.1) {\footnotesize$(\eps_1,\sigma_1)=(55\eps_0,\SI{1.2}{\siemens/\meter})$};
\draw[magenta,fill=magenta] (-0.4,-0.2) circle (0.1cm);
\node at (0,-0.4) {\footnotesize$(\eps_2,\sigma_2)=(45\eps_0,\SI{1.0}{\siemens/\meter})$};
\end{tikzpicture}
\caption{\label{IllustrationAnomalies}Illustration of the background (left), single (center) and multiple small anomalies (right).}
\end{center}
\end{figure}
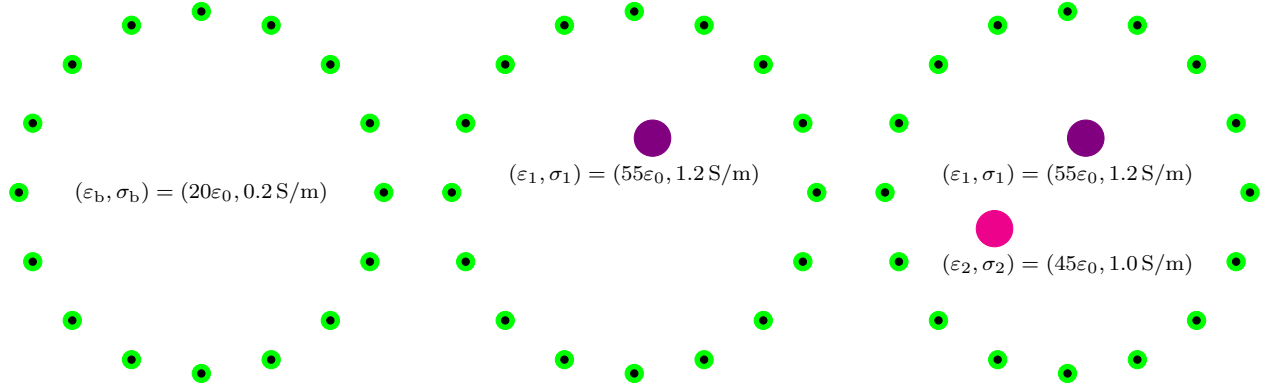

\begin{example}[Application of inaccurate background permeability]\label{EX-MU}
First, we assume that only the true value of $\mub$ is unknown, that is, we applied alternative wavenumber $\ka$ that satisfies
\[\ka^2=\omega^2\mua\left(\epsb+i\frac{\sigmab}{\omega}\right).\]
Then, identified location of anomaly becomes
\[\mr=\left(\frac{\kb}{\ka}\right)\mr_\star=\sqrt{\frac{\mub}{\mua}}\mr_\star.\]
Hence, the identified location will approach the origin if $\mua>\mub$ and. Otherwise, the identified location will be far from the origin if $\mua<\mub$.

Figure \ref{Result_MU1} shows maps of $\mathfrak{F}(\ka,\mr)$ with various $\ka$ in the presence of $D_1$. As we already mentioned, as the value of $\mua$ increases, the identified location approaches the origin. Otherwise, as the value of $\mua$ decreases, the identified location becomes far from the origin. It is interesting to examine the size of the identified anomaly becomes small and large as $\mua$ increases and decreases, respectively. We can observe the same phenomenon in the presence of multiple anomalies $D_1$ and $D_2$, as shown in Figure \ref{Result_MU2}. 
\end{example}

\begin{figure}[h]
\begin{center}
\subfigure[$\mua=\mub$]{\includegraphics[width=0.33\textwidth]{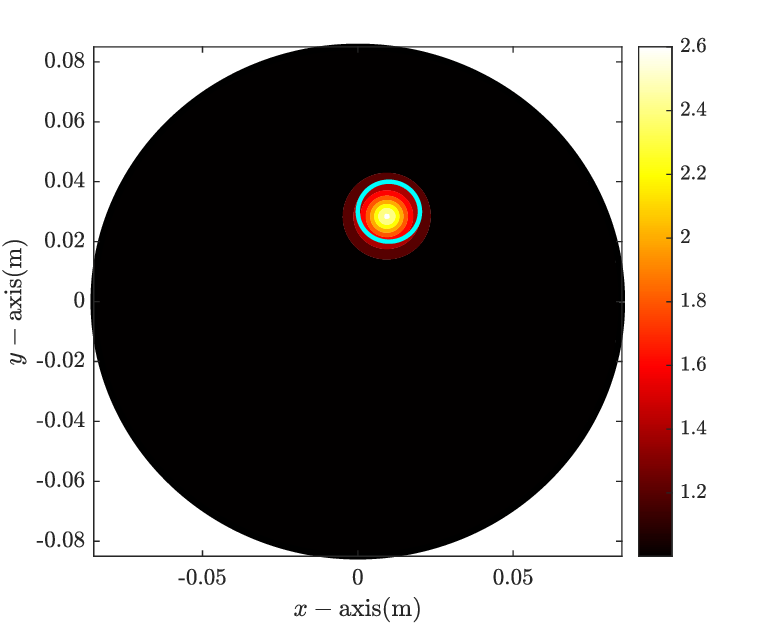}}\hfill
\subfigure[$\mua=2\mub$]{\includegraphics[width=0.33\textwidth]{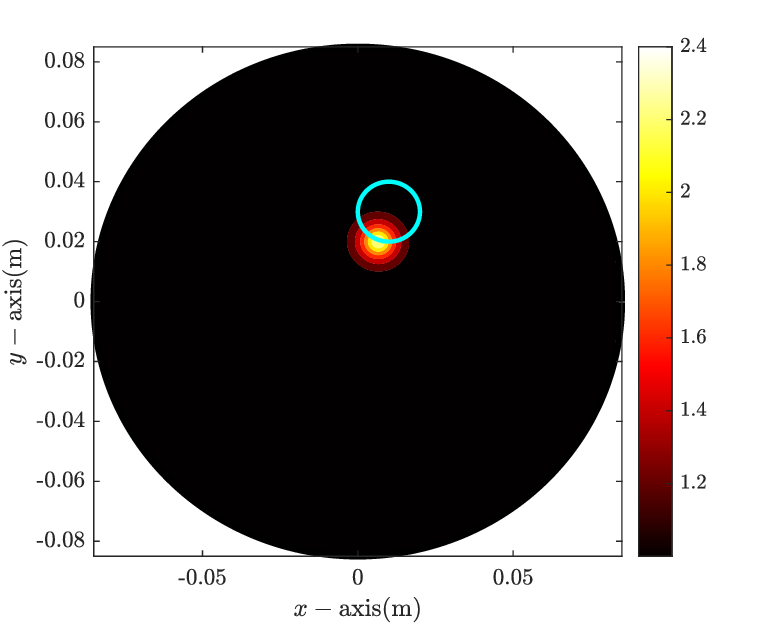}}\hfill
\subfigure[$\mua=10\mub$]{\includegraphics[width=0.33\textwidth]{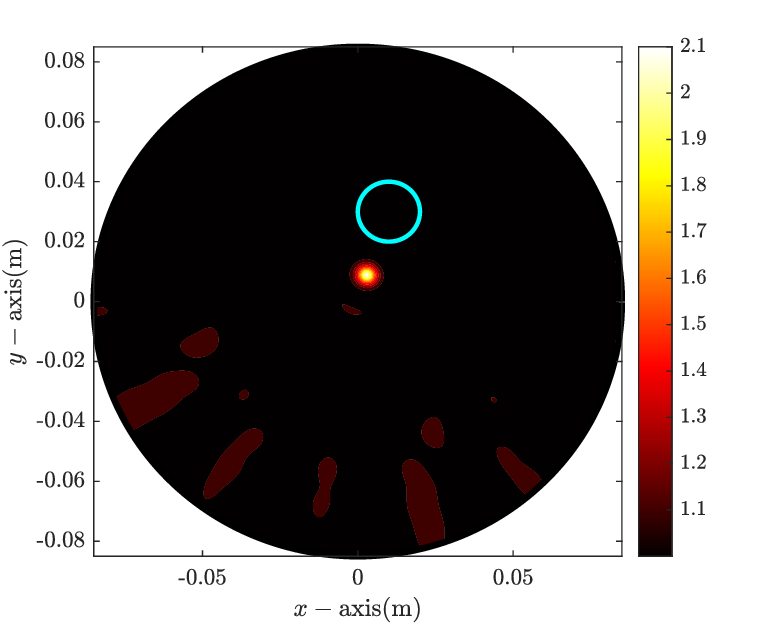}}\\
\subfigure[$\mua=0.5\mub$]{\includegraphics[width=0.33\textwidth]{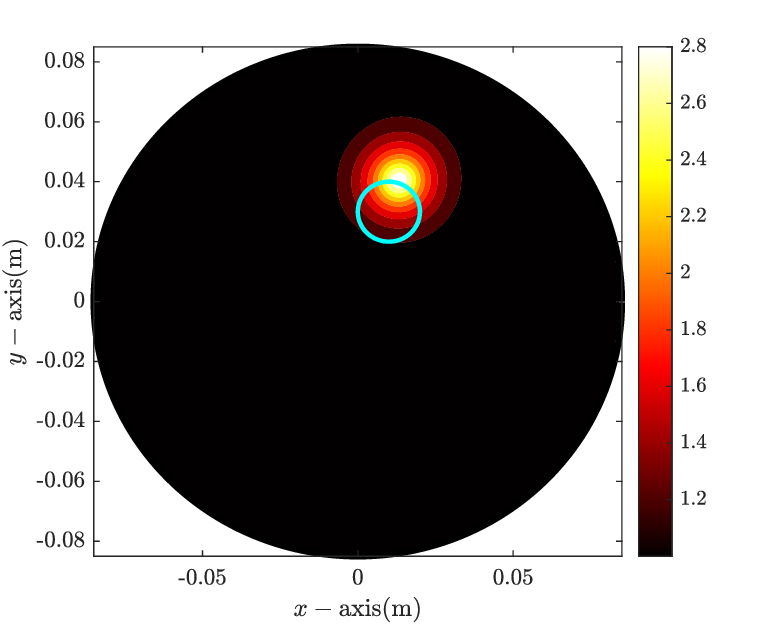}}\hfill
\subfigure[$\mua=0.2\mub$]{\includegraphics[width=0.33\textwidth]{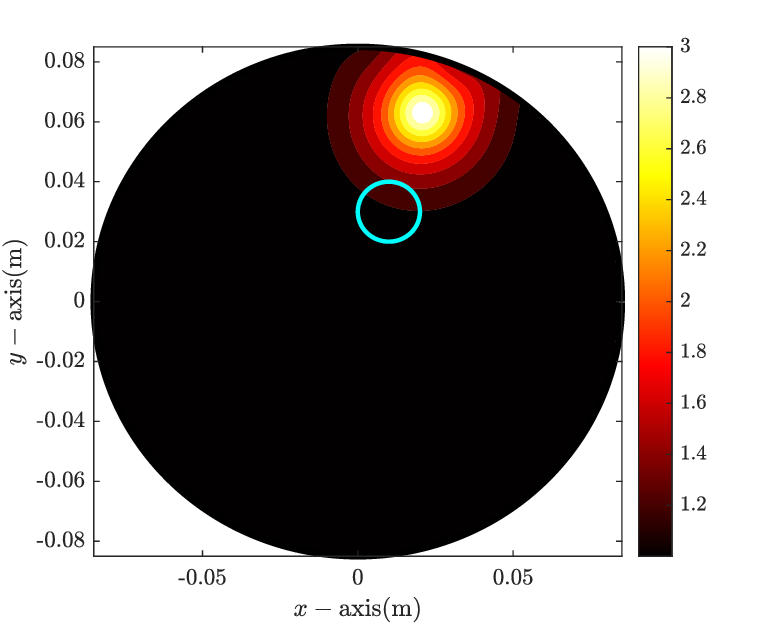}}
\subfigure[$\mua=0.1\mub$]{\includegraphics[width=0.33\textwidth]{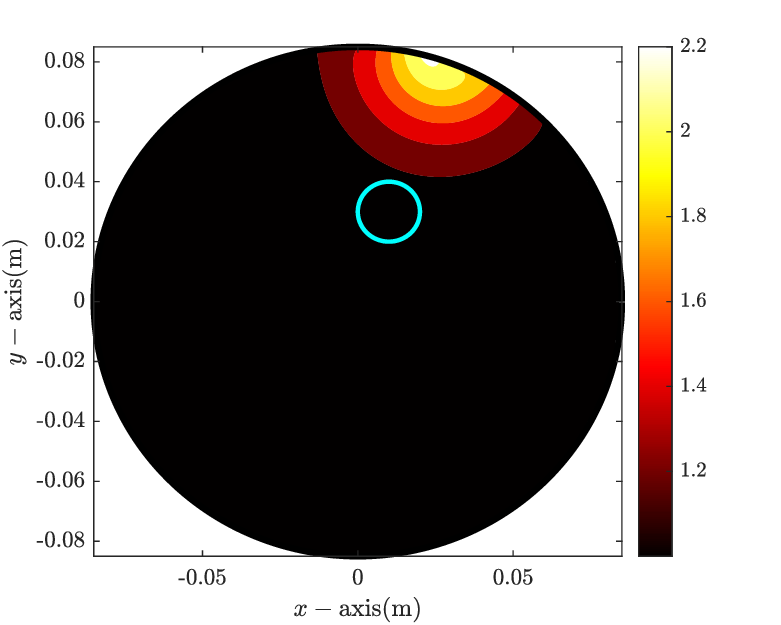}}
\caption{\label{Result_MU1} (Example \ref{EX-MU}) Maps of $\mathfrak{F}(\ka,\mr)$. Cyan-colored circle describes the boundary of anomaly.}
\end{center}
\end{figure}

\begin{figure}[h]
\begin{center}
\subfigure[$\mua=\mub$]{\includegraphics[width=0.33\textwidth]{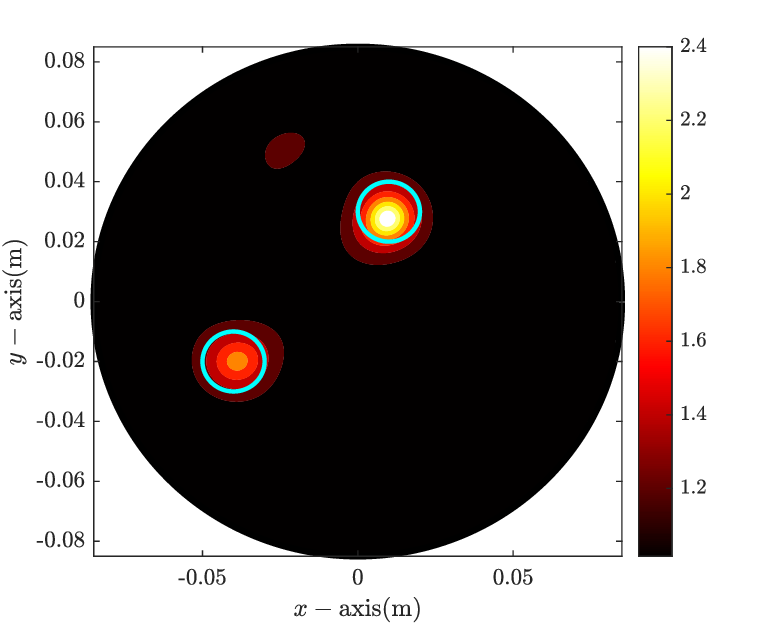}}\hfill
\subfigure[$\mua=2\mub$]{\includegraphics[width=0.33\textwidth]{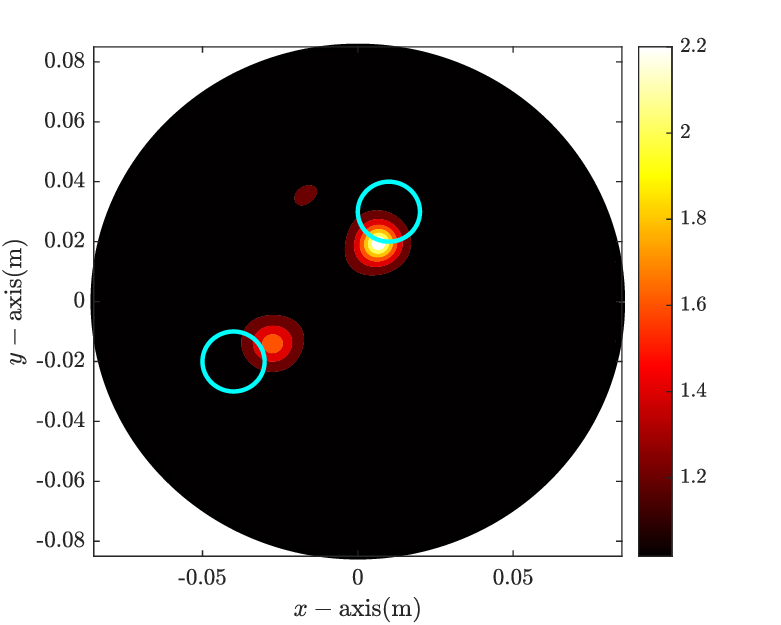}}\hfill
\subfigure[$\mua=10\mub$]{\includegraphics[width=0.33\textwidth]{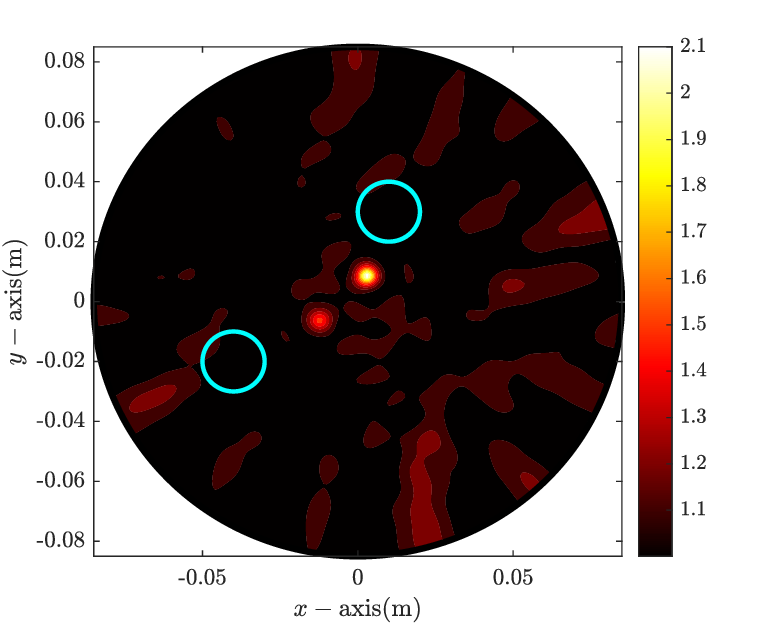}}\\
\subfigure[$\mua=0.5\mub$]{\includegraphics[width=0.33\textwidth]{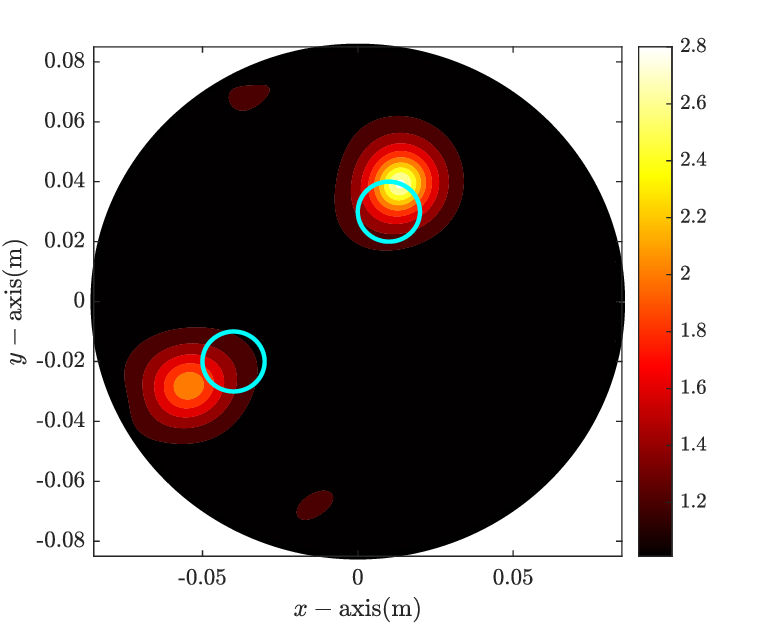}}\hfill
\subfigure[$\mua=0.2\mub$]{\includegraphics[width=0.33\textwidth]{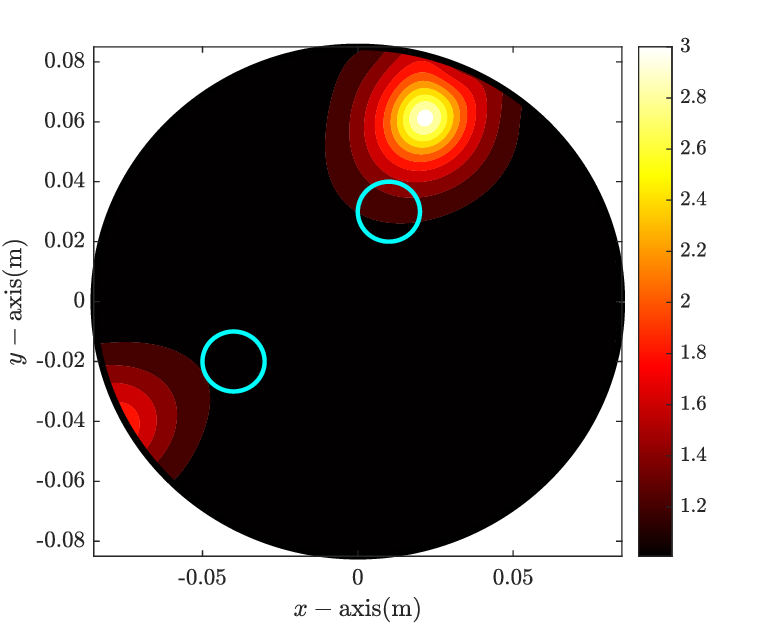}}
\subfigure[$\mua=0.1\mub$]{\includegraphics[width=0.33\textwidth]{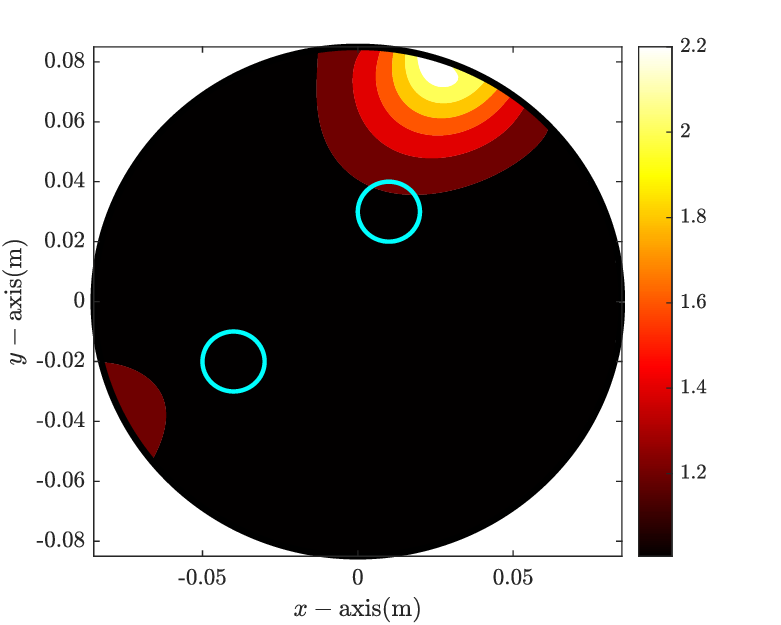}}
\caption{\label{Result_MU2} (Example \ref{EX-MU}) Maps of $\mathfrak{F}(\ka,\mr)$. Cyan-colored circles denote the boundary of anomalies.}
\end{center}
\end{figure}

\begin{example}[Application of inaccurate background permittivity]\label{EX-EPS}
Next, let us assume that only the true value of $\epsb$ is unknown, that is, we applied alternative wavenumber $\ka$ that satisfies
\[\ka^2=\omega^2\mub\left(\epsa+i\frac{\sigmab}{\omega}\right).\]
Note that, if $\epsa$ satisfies the condition \eqref{Condition}, the identified location will be
\[\mr=\left(\frac{\kb}{\ka}\right)\mr_\star=\sqrt{\frac{\omega\epsb+i\sigmab}{\omega\epsa+i\sigmab}}\mr_\star\approx\sqrt{\frac{\epsb}{\epsa}}\mr_\star.\]
Hence, similar to the Example \ref{EX-MU}, the identified location will approach and be far from the origin if $\epsa>\epsb$ and $\epsa<\epsb$, respectively.

Figure \ref{Result_EPS1} shows maps of $\mathfrak{F}(\ka,\mr)$ with various $\ka$ in the presence of $D_1$. Similar to the results in Figure \ref{Result_MU1}, the identified location approaches the origin as the value of $\epsa$ increases. Otherwise, the identified location becomes far from the origin as the value of $\epsa$ decreases. It is interesting to examine the size of the identified anomaly becomes small and large as $\epsa$ increases and decreases, respectively. We can observe the same phenomenon in the presence of multiple anomalies $D_1$ and $D_2$, as shown in Figure \ref{Result_EPS2}. 
\end{example}

\begin{figure}[h]
\begin{center}
\subfigure[$\epsa=\epsb$]{\includegraphics[width=0.33\textwidth]{True1}}\hfill
\subfigure[$\epsa=2\epsb$]{\includegraphics[width=0.33\textwidth]{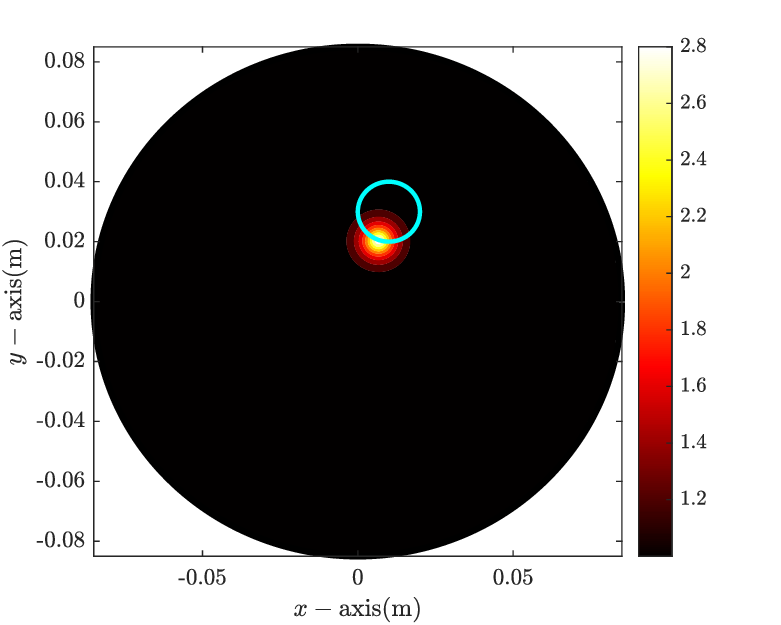}}\hfill
\subfigure[$\epsa=10\epsb$]{\includegraphics[width=0.33\textwidth]{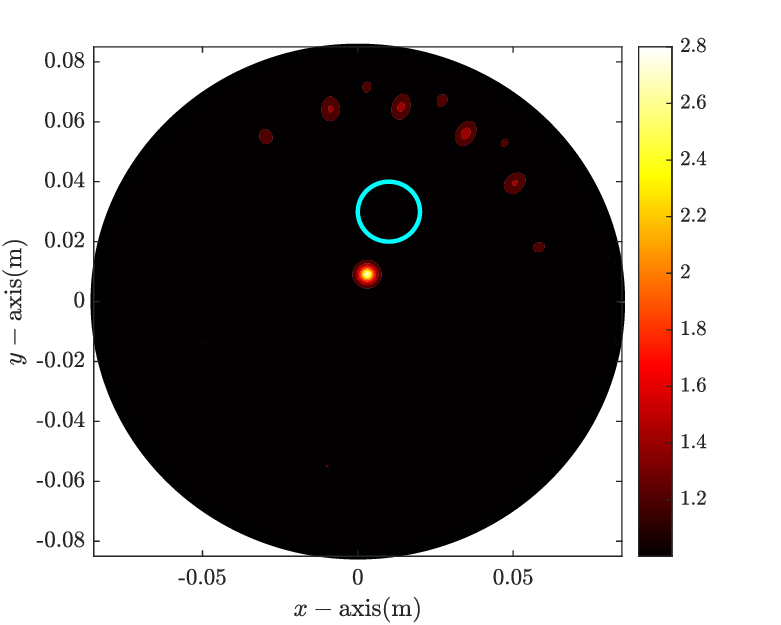}}\\
\subfigure[$\epsa=0.5\epsb$]{\includegraphics[width=0.33\textwidth]{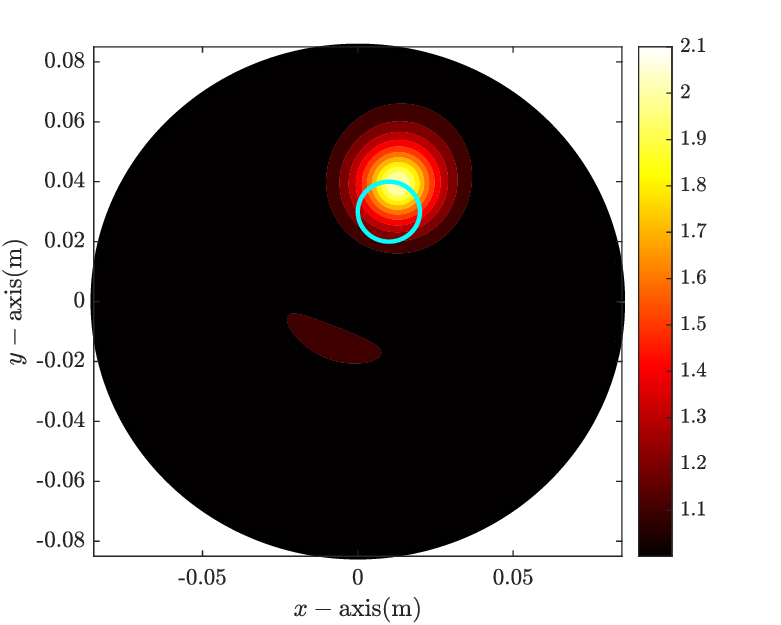}}\hfill
\subfigure[$\epsa=0.2\epsb$]{\includegraphics[width=0.33\textwidth]{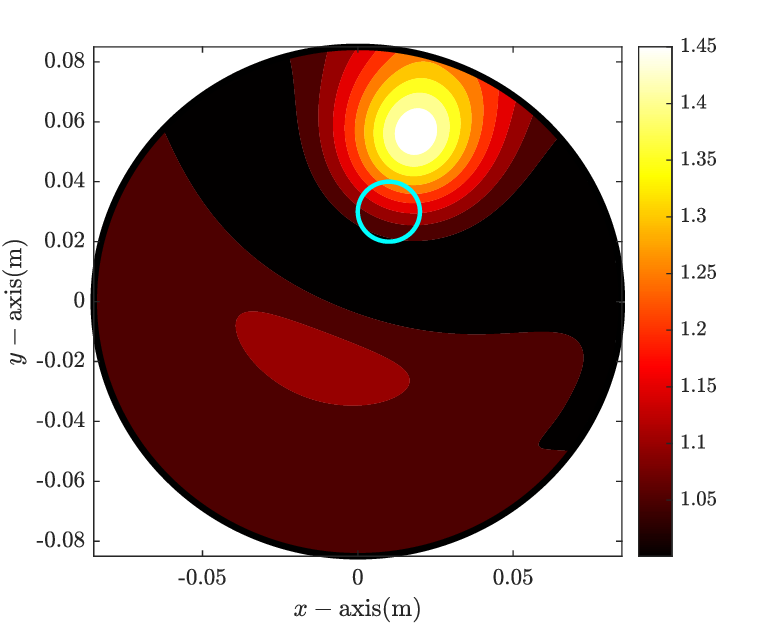}}
\subfigure[$\epsa=0.1\epsb$]{\includegraphics[width=0.33\textwidth]{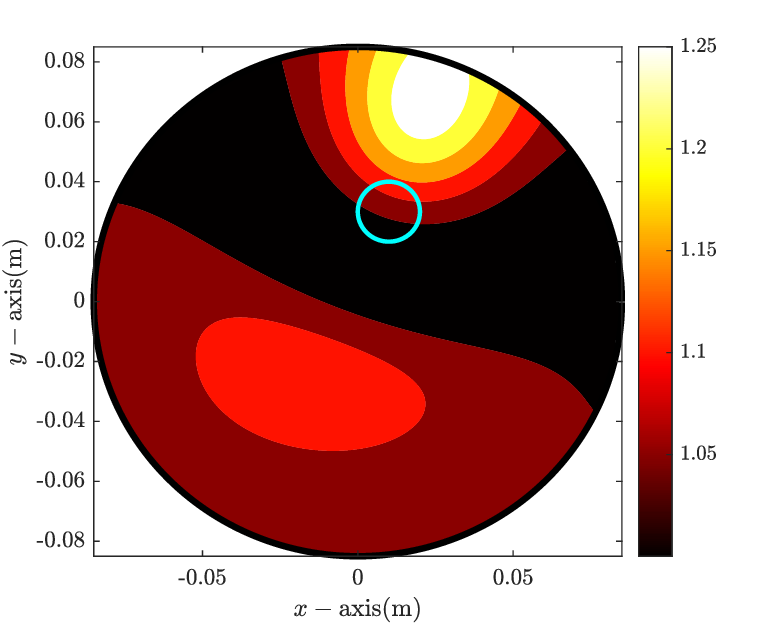}}
\caption{\label{Result_EPS1} (Example \ref{EX-EPS}) Maps of $\mathfrak{F}(\ka,\mr)$. Cyan-colored circle describes the boundary of anomaly.}
\end{center}
\end{figure}

\begin{figure}[h]
\begin{center}
\subfigure[$\epsa=\epsb$]{\includegraphics[width=0.33\textwidth]{True2}}\hfill
\subfigure[$\epsa=2\epsb$]{\includegraphics[width=0.33\textwidth]{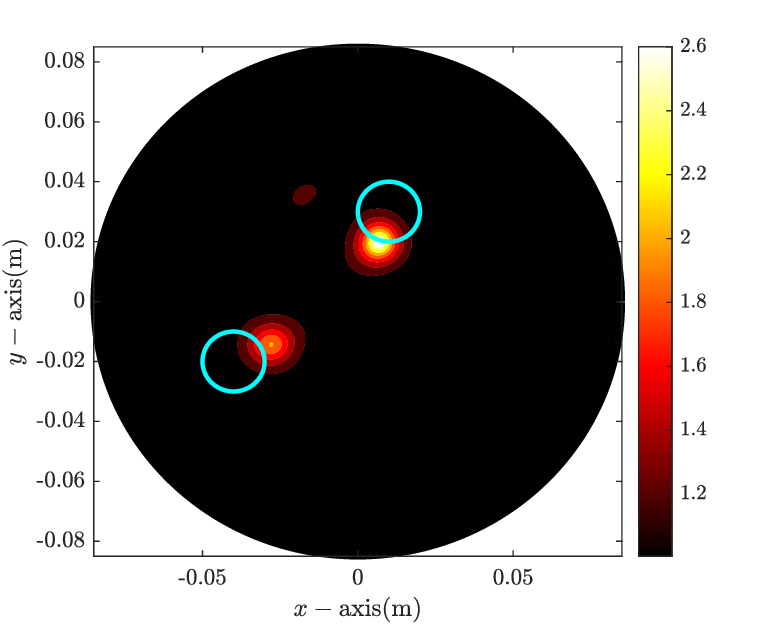}}\hfill
\subfigure[$\epsa=10\epsb$]{\includegraphics[width=0.33\textwidth]{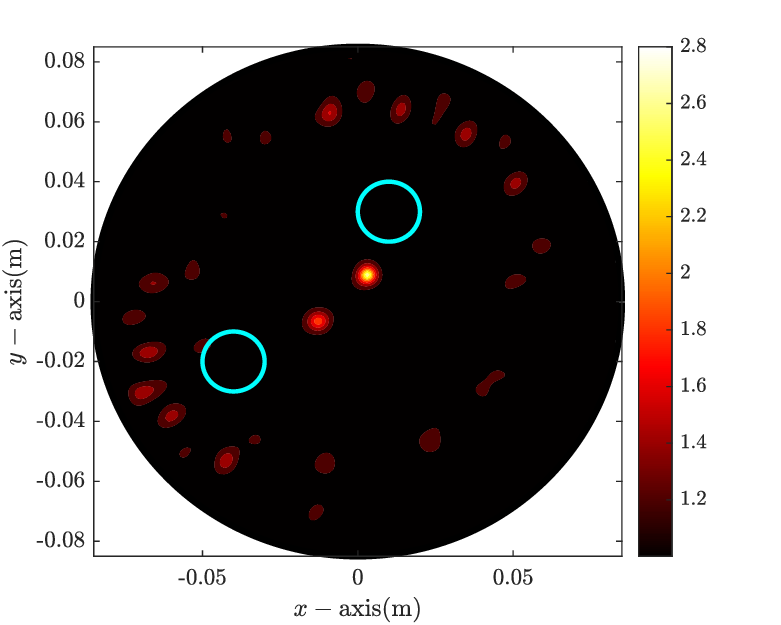}}\\
\subfigure[$\epsa=0.5\epsb$]{\includegraphics[width=0.33\textwidth]{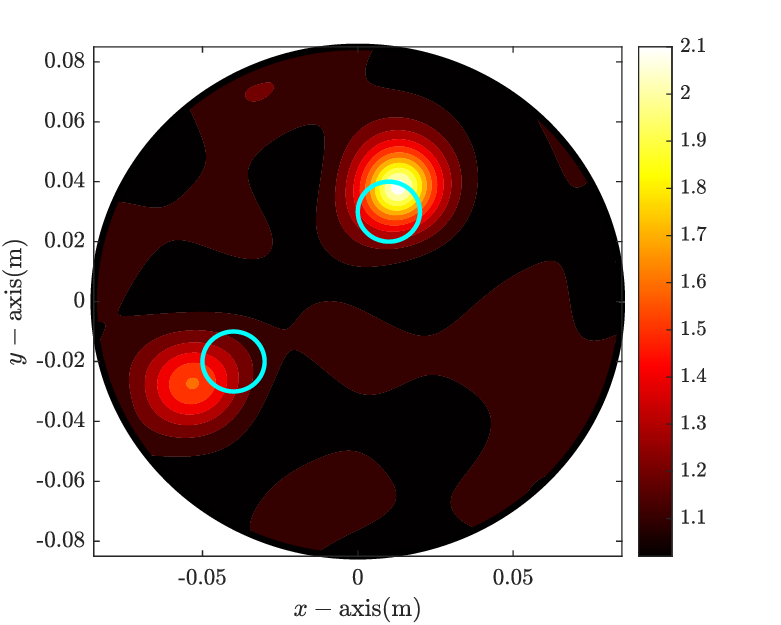}}\hfill
\subfigure[$\epsa=0.2\epsb$]{\includegraphics[width=0.33\textwidth]{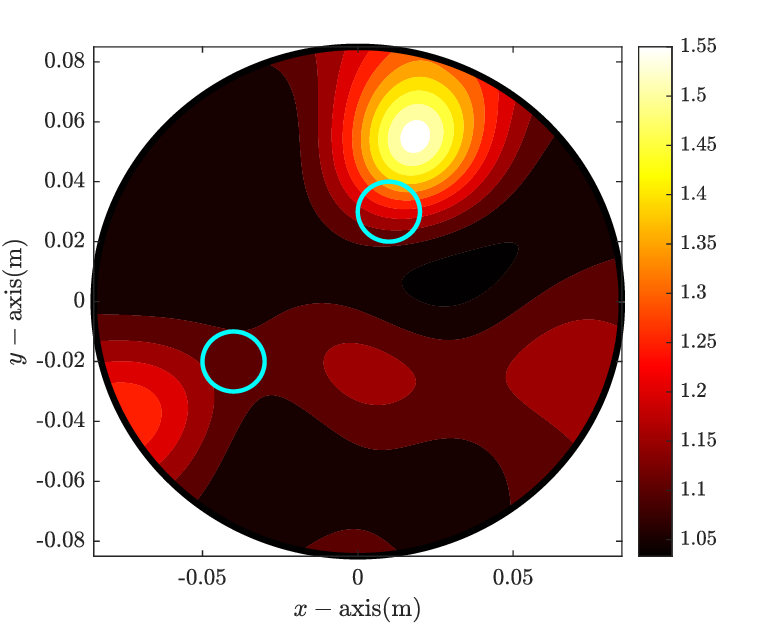}}
\subfigure[$\epsa=0.1\epsb$]{\includegraphics[width=0.33\textwidth]{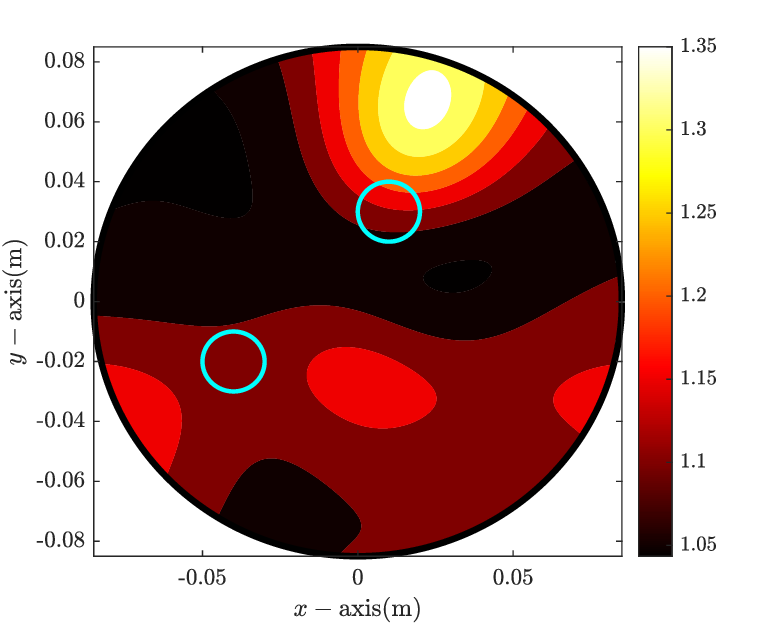}}
\caption{\label{Result_EPS2} (Example \ref{EX-EPS}) Maps of $\mathfrak{F}(\ka,\mr)$. Cyan-colored circles denote the boundary of anomalies.}
\end{center}
\end{figure}

\begin{example}[Application of inaccurate background conductivity]\label{EX-SIGMA}
Here, we consider the case when only the true value of $\sigmab$ is unknown and correspondingly, we apply an alternative wavenumber $\ka$ that satisfies
\[\ka^2=\omega^2\mub\left(\epsb+i\frac{\sigmaa}{\omega}\right).\]
Notice that opposite to the Examples \ref{EX-MU} and \ref{EX-EPS}, if $\sigmaa$ satisfies the condition \eqref{Condition}, the identified location will be
\[\mr=\left(\frac{\kb}{\ka}\right)\mr_\star=\sqrt{\frac{\omega\epsb+i\sigmab}{\omega\epsb+i\sigmaa}}\mr_\star\approx\sqrt{\frac{\omega\epsb}{\omega\epsb}}\mr_\star=\mr_\star.\]
Hence, it can be expected that almost exact location of the anomaly can be identified through the map of $\mathfrak{F}(\ka,\mr)$. This means that, it will be possible to identify the location of the anomaly by selecting a small value of $\sigmaa$ although its true value is unknown. Otherwise, if $\sigmaa$ does not satisfy the condition \eqref{Condition}, it will be impossible to identify the anomaly because the Born approximation cannot be applied to design the imaging function.

Figure \ref{Result_SIGMA1} shows maps of $\mathfrak{F}(\ka,\mr)$ with various $\ka$ in the presence of $D_1$. As we discussed previously, it is possible to identify almost exact location of $D_1$ if the value of $\sigmaa$ is sufficiently small. However, owing to the appearance of unexpected artifacts with large magnitudes, it is very difficult to recognize the location of $D_1$ if the value of $\sigmaa$ is not small. Hence, contrary to Example \ref{EX-EPS}, selecting a small value of $\sigmaa$ will guarantee successful identification of the anomaly without accurate value of $\sigmab$. We can observe the same phenomea in the presence of multiple anomalies, as shown in Figure \ref{Result_SIGMA2}.
\end{example}

\begin{figure}[h]
\begin{center}
\subfigure[$\sigmaa=\sigmab$]{\includegraphics[width=0.33\textwidth]{True1}}\hfill
\subfigure[$\sigmaa=2\sigmab$]{\includegraphics[width=0.33\textwidth]{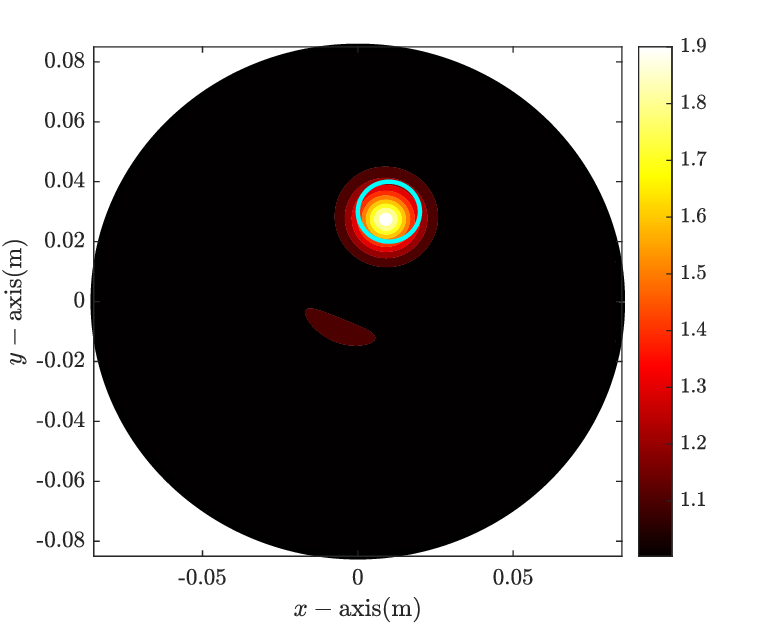}}\hfill
\subfigure[$\sigmaa=10\sigmab$]{\includegraphics[width=0.33\textwidth]{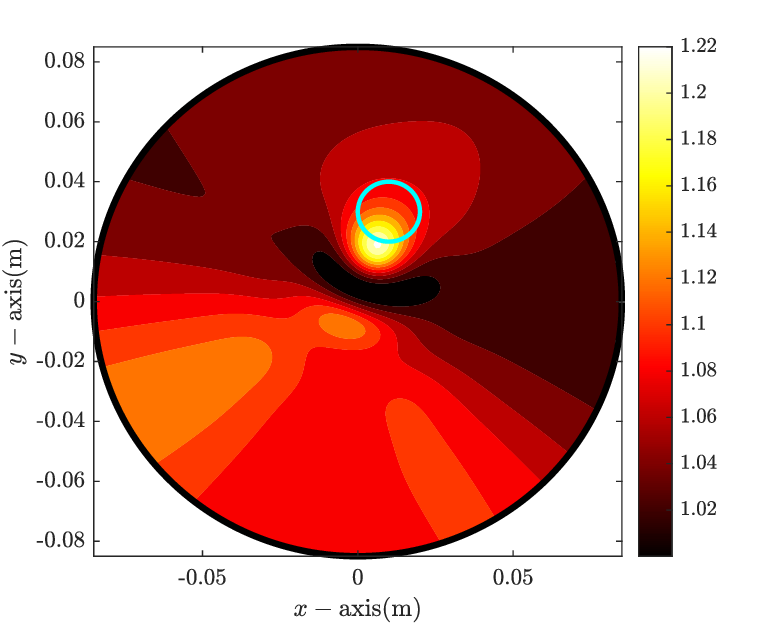}}\\
\subfigure[$\sigmaa=20\sigmab$]{\includegraphics[width=0.33\textwidth]{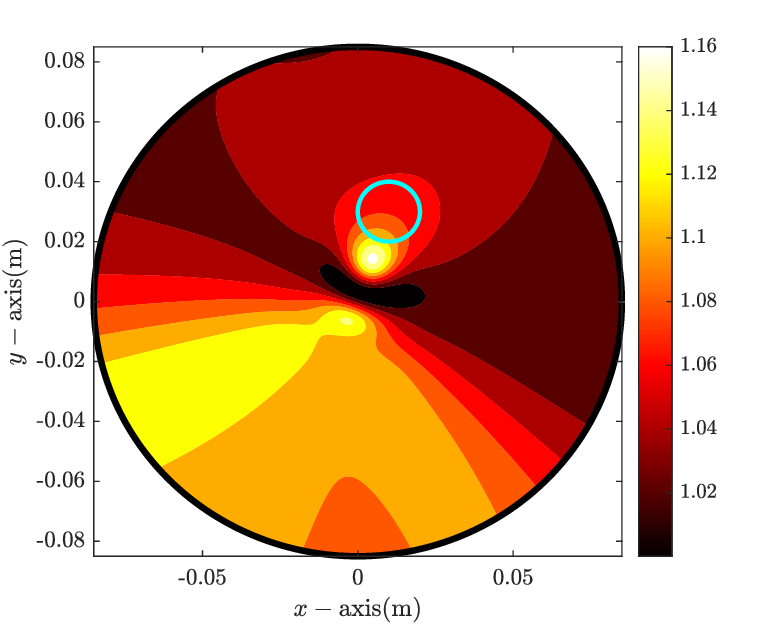}}\hfill
\subfigure[$\sigmaa=0.2\sigmab$]{\includegraphics[width=0.33\textwidth]{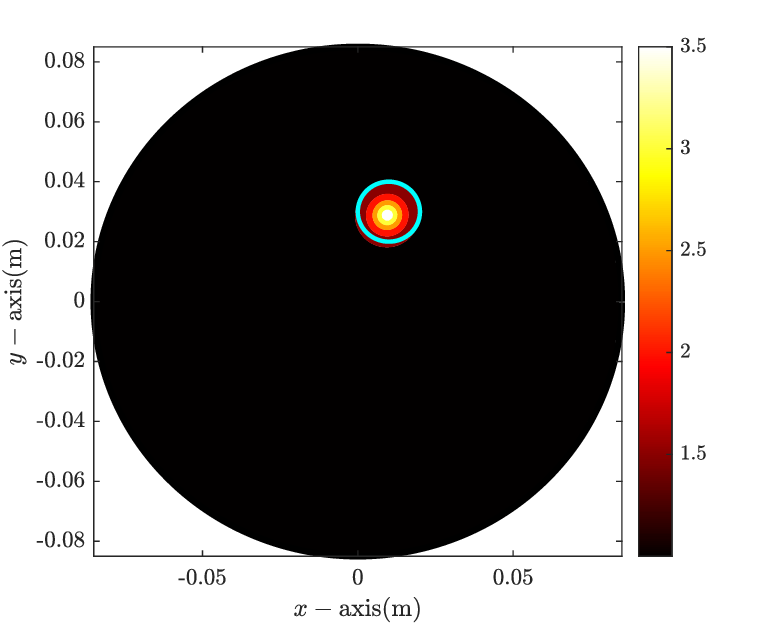}}\hfill
\subfigure[$\sigmaa=0.1\sigmab$]{\includegraphics[width=0.33\textwidth]{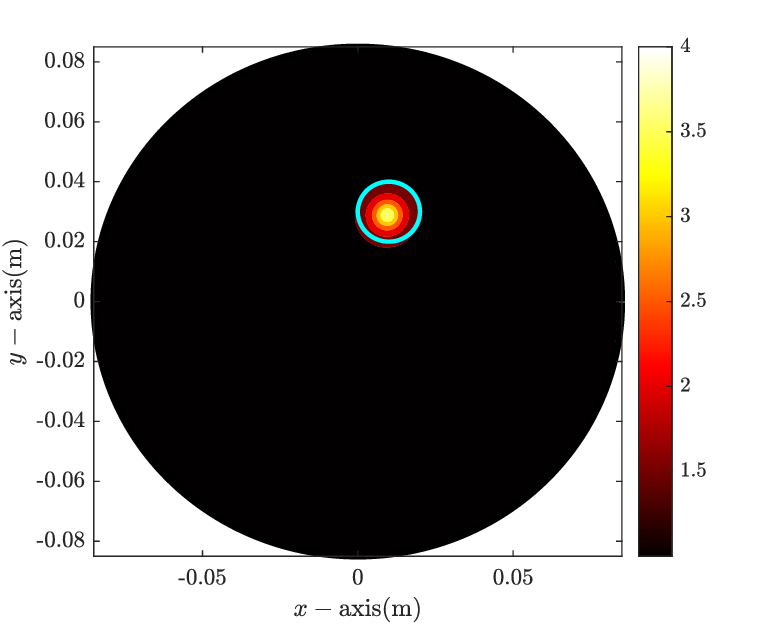}}
\caption{\label{Result_SIGMA1} (Example \ref{EX-SIGMA}) Maps of $\mathfrak{F}(\ka,\mr)$. Cyan-colored circle describes the boundary of anomaly.}
\end{center}
\end{figure}

\begin{figure}[h]
\begin{center}
\subfigure[$\sigmaa=\sigmab$]{\includegraphics[width=0.33\textwidth]{True2}}\hfill
\subfigure[$\sigmaa=2\sigmab$]{\includegraphics[width=0.33\textwidth]{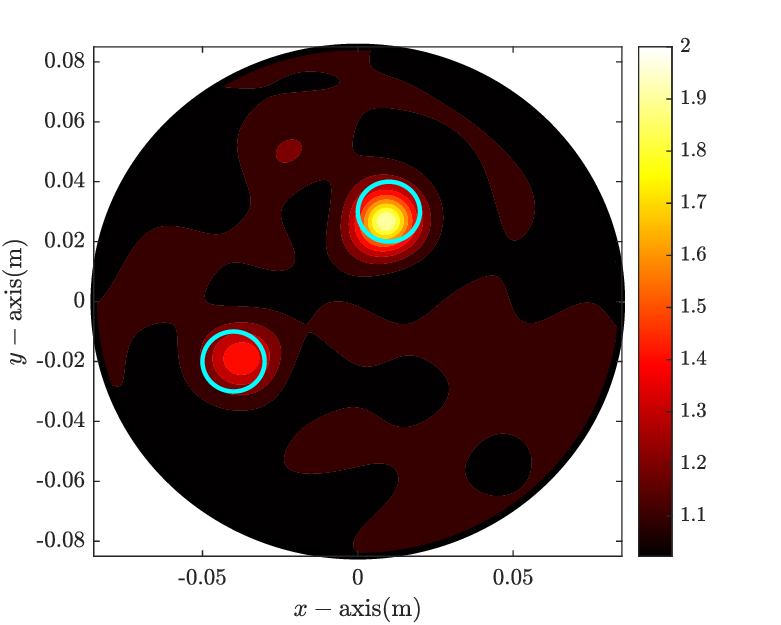}}\hfill
\subfigure[$\sigmaa=10\sigmab$]{\includegraphics[width=0.33\textwidth]{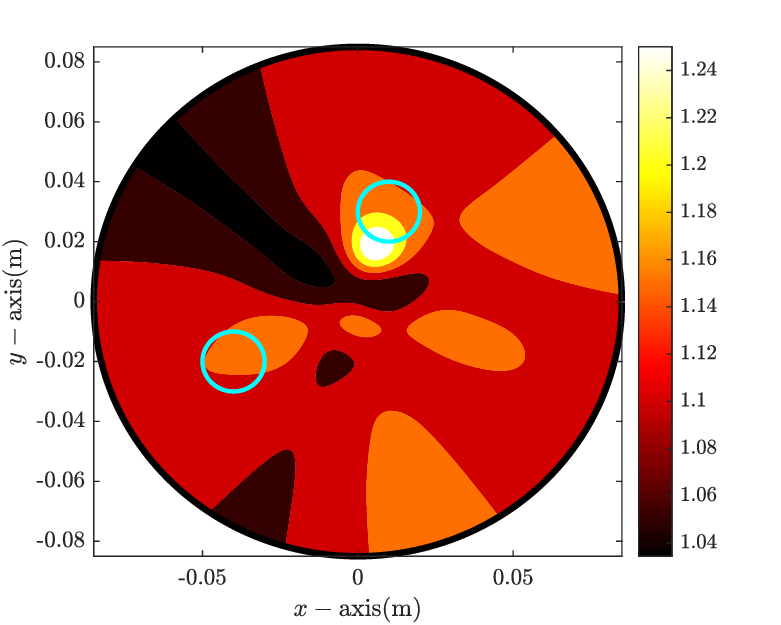}}\\
\subfigure[$\sigmaa=20\sigmab$]{\includegraphics[width=0.33\textwidth]{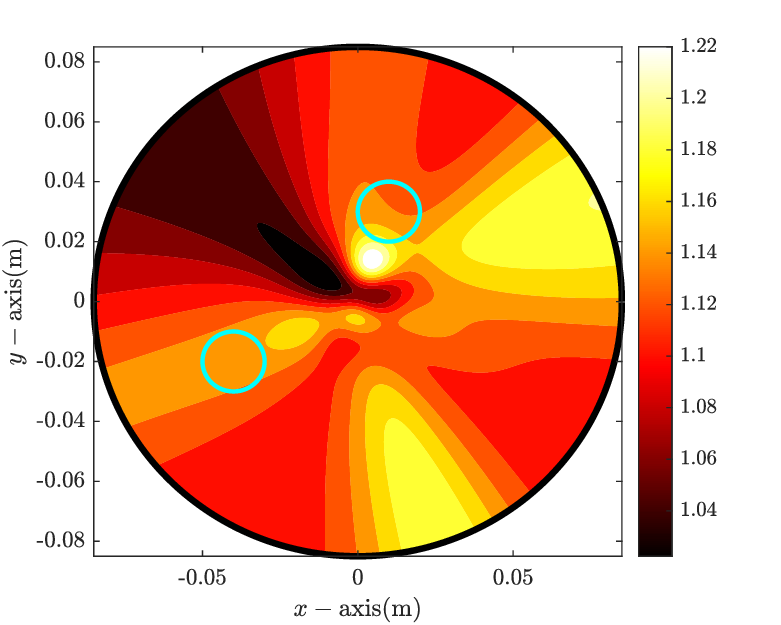}}\hfill
\subfigure[$\sigmaa=0.2\sigmab$]{\includegraphics[width=0.33\textwidth]{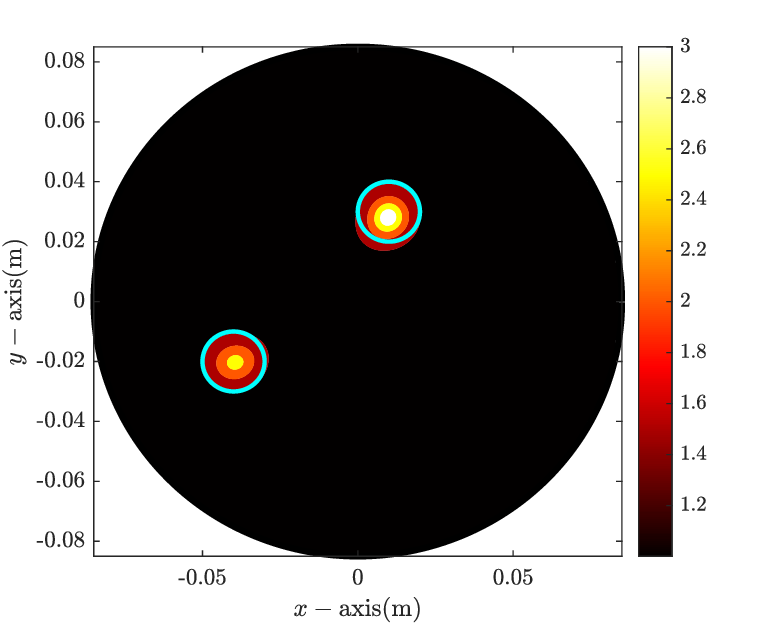}}\hfill
\subfigure[$\sigmaa=0.1\sigmab$]{\includegraphics[width=0.33\textwidth]{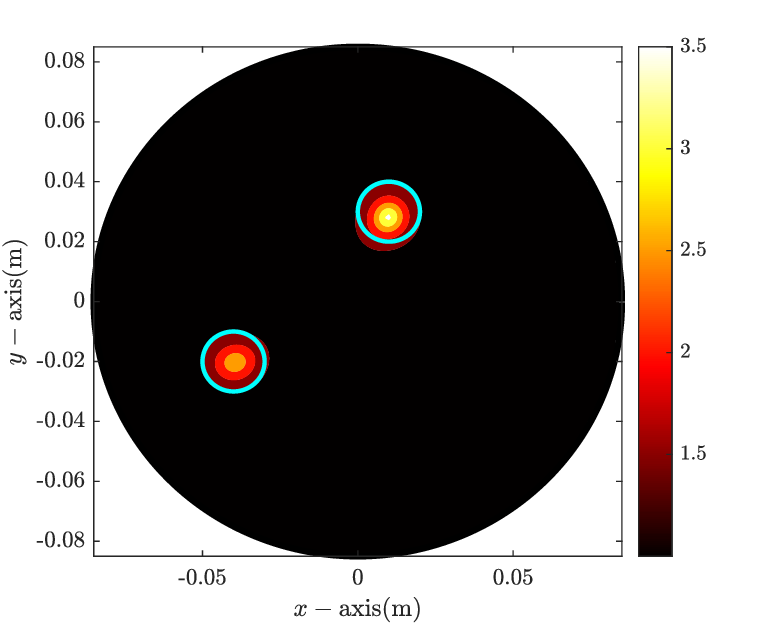}}
\caption{\label{Result_SIGMA2} (Example \ref{EX-SIGMA}) Maps of $\mathfrak{F}(\ka,\mr)$. Cyan-colored circles describe the boundary of anomalies.}
\end{center}
\end{figure}

\section{Conclusion}\label{sec:5}
Based on the integral equation for the scattered-field $S-$parameter and singular value decomposition of the scattering matrix in the presence of a small anomaly, we showed that the imaging function of MUSIC can be expressed by an infinite series of Bessel functions and applied wavenumber. Thanks to the theoretical result, we confirmed why an inaccurate location of the anomaly was retrieved when inaccurate value of background permeability, permittivity, or conductivity. However, the relationship between the retrieved size of the anomaly and the applied wavenumber remains unknown. It will be interesting to investigate a mathematical theory to explain this phenomenon. Moreover, the development of an effective algorithm for estimating exact value of background wavenumber would be a valuable addition to this work.

\section*{Acknowledgments}
This research was supported by the National Research Foundation of Korea (NRF) grant funded by the Korea government (MSIT) (NRF-2020R1A2C1A01005221).

\end{document}